\documentclass[11pt]{article}
\pdftrailerid{}
\usepackage[T1]{fontenc}
\usepackage{lmodern}
\usepackage[a4paper,margin=28mm]{geometry}
\usepackage{amsmath,amssymb,amsthm,booktabs,tabularx,array}
\usepackage[hidelinks]{hyperref}
\hypersetup{pdftitle={Minimum central circles: chain order, seam obstructions, and exact finite certification},pdfauthor={Maurizio Falconi}}
\newtheorem{theorem}{Theorem}[section]
\newtheorem{lemma}[theorem]{Lemma}

\newtheorem{corollary}[theorem]{Corollary}
\theoremstyle{definition}

\theoremstyle{remark}

\newcommand{\Rchain}{R_{\mathrm{chain}}}
\newcommand{\Rfull}{R_{\mathrm{full}}}
\newcommand{\Rstar}{R^*}
\newcommand{\snapshot}{6c16af422d1cb38641c62d43b6e0e547921b9ba9}
\title{Minimum central circles:\\chain order, seam obstructions,\\and exact finite certification}
\author{Maurizio Falconi\\[3pt]\small Independent researcher, Parma, Italy\\
\small Corresponding author: \texttt{maurizio.falconi47@gmail.com}}
\date{September 19, 2026}
\begin{document}
\maketitle

\begin{abstract}
We minimize the radius of a central circle externally tangent to pairwise
non-overlapping circles of radii $1,2,\ldots,n$. We first study the
adjacent-chain relaxation, which imposes only the constraints between
consecutive circles. Its angular kernel has strict anti-Monge structure, so
Supnick's classical theorem supplies a fixed minimizing cyclic order at every
central radius. A monotonicity argument transfers this order to the
chain-radius problem. Full geometric feasibility requires every pairwise
angular constraint. We give a self-contained criterion for feasibility of the
canonical Supnick necklace on $\{k,\ldots,n\}$ at its chain root, expressed
by one seam deficit. For every fixed $k$ the seam eventually fails and its
failure persists. We prove the complete classifications for $k=1,2,3$, whose
first failures occur at $n=8,13,17$, respectively.

For the global infimum, we certify rational brackets
$L_n<R^*(n)\le U_n$ of exact width $10^{-11}$
for every $3\le n\le14$. The lower proof combines independently enclosed
angles, integer dynamic programming, and complete cyclic-order coverage;
the upper proof uses exact rational Cartesian witnesses. The certificate
does not depend on floating-point pruning in the numerical search. We
distinguish these
finite brackets and fixed-order theorems from numerical descriptions of
floating circles and from unresolved claims about all optimal placements.
\end{abstract}



\section{Introduction and contributions}
Consider $n\ge3$ circles with radii $1,2,\ldots,n$, each externally tangent
to a central circle. Their interiors must be pairwise disjoint, and the
central radius is to be minimized. Denote its global infimum by $\Rstar(n)$.
The freedom to permute the surrounding circles makes this both a geometric
feasibility problem and a discrete optimization problem. In particular,
an order minimizing the angle needed by consecutive circles need not
admit a non-overlapping placement at its chain-closing radius.

The problem and a conjectured ``pyramid'' order were posed by Dan in
January 2023 on Mathematics Stack Exchange
(\href{https://math.stackexchange.com/q/4619480}{question 4619480}).
The present paper gives a finite, standalone treatment. Its main certified
numerical result is an enclosure of the global optimum, not an assertion
that a printed decimal is an exact optimum:
\begin{equation}\label{eq:mainclaim}
 L_n<\Rstar(n)\le U_n,\qquad U_n-L_n=\frac1{10^{11}},
 \qquad n=3,\ldots,14.
\end{equation}
All endpoints are stated explicitly in Section~\ref{sec:results}.

The paper makes four contributions. First, the geometric angular kernel is
strictly anti-Monge for every central radius. Applying the
classical extreme-tour theorem of Supnick~\cite{supnick} and then comparing
closure roots solves the adjacent-chain relaxation, including for arbitrary
distinct positive radii. Second, we formulate full fixed-order feasibility
using both directed angular constraints for every pair.

Third, a single seam deficit characterizes feasibility at the chain root
of the canonical order on $\{k,\ldots,n\}$. Its eventual persistent failure is proved for
each fixed $k$; the complete strict classifications proved here are precisely
those for $k=1,2,3$.

Fourth, an exact finite certificate proves \eqref{eq:mainclaim} by combining
exhaustive lower-bound coverage with
rational geometric existence witnesses.

The new geometric conclusions are not part of Supnick's tour theorem.
That theorem supplies a minimizing order for a cost matrix with the specified
structure. The identification of that structure in the angular kernel,
the transfer from fixed radius to closure radius, the failure of the chain
relaxation as full geometry, and the subsequent exact global certificates
are separate steps. In particular, neither the chain theorem nor the seam
theorem establishes floating behavior at global optima.

The Monge sign conventions and well-solvable traveling salesman problems
are surveyed in~\cite{survey}. The all-pairs formulation is a system of
difference constraints, as in temporal constraint networks~\cite{stn}.
Related circle-packing work prescribes tangency structure. Mathews and
Zymaris~\cite{mathews-zymaris} generalize Descartes' circle theorem to
$n$-flowers, whose petals are tangent to the central circle and to their
cyclic neighbors. Collins and Stephenson~\cite{collins-stephenson} compute
circle packings realizing a specified pattern of tangencies. In Ringmin,
the outer radii are fixed, the cyclic order must be chosen, and full
pairwise non-overlap must be imposed without assuming tangency between
consecutive outer circles.
Container packing has different boundary conditions: here every surrounding
circle must touch the same central circle, and angular order is a primary
variable. We use no general packing theorem to infer global optimality.

An earlier finite preprint is available as
\href{https://arxiv.org/abs/2607.28654v2}{arXiv:2607.28654v2}.
The separate asymptotic study,
\emph{Minimum central circles: an effective characterization of the global
asymptotic constant}, is available as
\href{https://arxiv.org/abs/2609.13630}{arXiv:2609.13630}.
No asymptotic theorem or conjectural asymptotic coefficient is needed for
any result proved here.

\section{Geometric angular reformulation}\label{sec:geometry}
Place the center of the central circle at the origin. An outer circle
of radius $a>0$
has center $(R+a)(\cos\phi_a,\sin\phi_a)$. For $R,a,b>0$ put
\begin{equation}\label{eq:kernel}
 \theta_R(a,b)=2\arcsin\sqrt{\frac{ab}{(R+a)(R+b)}}\in(0,\pi).
\end{equation}
\begin{lemma}[Both directed arcs]\label{lem:angular}
If the counterclockwise angular separation of two outer centers is
$A\in[0,2\pi]$, their circles have disjoint interiors if and only if
\begin{equation}\label{eq:arcs}
 \theta_R(a,b)\le A\le2\pi-\theta_R(a,b).
\end{equation}
The kernel is symmetric, strictly decreasing in $R$, and strictly increasing
in each outer radius.
\end{lemma}
\begin{proof}
The squared center distance is
$(a-b)^2+4(R+a)(R+b)\sin^2(A/2)$. Comparison with $(a+b)^2$
is equivalent to $\cos A\le1-2ab/((R+a)(R+b))$, which proves
\eqref{eq:arcs} on the full directed interval. The square-root argument in
\eqref{eq:kernel} lies in $(0,1)$, and its monotonicities prove the rest.
\end{proof}

For a cyclic order $\sigma=(s_0,\ldots,s_{N-1})$ of at least three
positive radii define
\[
 C_\sigma(R)=\sum_{i=0}^{N-1}\theta_R(s_i,s_{i+1}),\qquad s_N=s_0.
\]
It decreases continuously from $N\pi>2\pi$ to zero. Thus
$C_\sigma(\Rchain(\sigma))=2\pi$ defines a unique positive
\emph{chain radius}. Let $\Rfull(\sigma)$ be the infimum of radii
admitting an all-pairs-feasible placement in this order. Summing adjacent
gap constraints, including the closing gap, gives
\begin{equation}\label{eq:hierarchy}
 \Rchain(\sigma)\le\Rfull(\sigma),\qquad
 \min_\sigma\Rchain(\sigma)\le\Rstar(n)=\min_\sigma\Rfull(\sigma).
\end{equation}
For sufficiently large $R$, equal angular gaps give a feasible placement,
since every $\theta_R(a,b)$ tends to zero. These infima are finite; none
of the certificate arguments below presumes attainment of an infimum.

\section{Anti-Monge structure and the Supnick chain order}\label{sec:chain}
Our sign convention calls a symmetric kernel anti-Monge when
\begin{equation}\label{eq:antimonge}
 \theta_R(a,b)+\theta_R(a',b')>
 \theta_R(a,b')+\theta_R(a',b)
 \quad(a<a',\ b<b').
\end{equation}
\begin{lemma}
The kernel \eqref{eq:kernel} satisfies \eqref{eq:antimonge} for every $R>0$.
\end{lemma}
\begin{proof}
Direct differentiation gives
\begin{equation}\label{eq:mixed}
 \frac{\partial^2\theta_R}{\partial a\,\partial b}
 =\frac{\sqrt R}{2\sqrt{ab}(R+a+b)^{3/2}}>0.
\end{equation}
Integrating over $[a,a']\times[b,b']$ proves the strict inequality.
Analytic diagonal values are allowed in this calculation; they are not
self-pair constraints in the geometric model.
\end{proof}

Define the relevant tour on $N$ ranked vertices as follows.
Set $H=\lceil N/2\rceil$. For successive $j=0,1,\ldots$,
append $1+2j$ to list $A$ when it is at most $H$, and then append
$N-1-2j$ when it is greater than $H$. Construct list $B$ by appending
$2+2j$ when it is at most $H$, followed by $N-2-2j$ when it is greater
than $H$. Stop each list when neither entry qualifies. Concatenate
$A$, the reverse of $B$, and $N$, and close the cycle. Every rank occurs
once. Denote this rank cycle by $\sigma^*_N$. For example,
$\sigma^*_4=(1,3,2,4)$ and $\sigma^*_7=(1,6,3,4,5,2,7)$.

For a symmetric Monge cost matrix, Supnick's extreme-tour theorem
identifies this rank cycle as a maximum-cost tour; see the maximum-tour
discussion in~\cite{survey}. Apply it to the Monge matrix $-\theta_R$.
The sign reversal makes $\sigma^*_N$ a minimum-angle tour. The classical
tour theorem is the only imported ordering theorem.

\begin{theorem}[Chain optimization]\label{thm:chain}
For arbitrary distinct radii $0<r_1<\cdots<r_N$, the order obtained by
replacing rank $i$ in $\sigma^*_N$ by $r_i$ minimizes $C_\sigma(R)$
at every $R>0$, and hence minimizes $\Rchain(\sigma)$. Its chain
radius is a lower bound for the global geometric optimum. If its cumulative
adjacent-angle placement is all-pairs feasible at that radius, the bound
is attained and the placement is globally optimal.
\end{theorem}
\begin{proof}
The fixed-radius assertion follows from \eqref{eq:antimonge} and Supnick's
theorem. At $R=\Rchain(\sigma)$ for any competing order, the minimizing
tour has angle sum at most $2\pi$. Strict decrease of its sum in $R$
places its root no higher than $\Rchain(\sigma)$. Now use
\eqref{eq:hierarchy}. An all-pairs-feasible placement at the lower bound
proves equality.
\end{proof}

\section{Fixed-order geometric feasibility}\label{sec:fixed}
Fix $R$ and $\sigma$, and unroll angles as
$\phi_0<\cdots<\phi_{N-1}<\phi_0+2\pi$. For every $i<j$ impose
\begin{equation}\label{eq:difference}
 \theta_R(s_i,s_j)\le\phi_j-\phi_i
 \le2\pi-\theta_R(s_i,s_j).
\end{equation}
The positivity and upper bound already enforce the indicated order and
wrap interval. By Lemma~\ref{lem:angular}, this system is equivalent to
full geometric feasibility. Written as $x_v-x_u\le c_{uv}$ it is
feasible exactly when its directed constraint graph has no negative cycle.
Necessity follows by summing around a cycle; for sufficiency, distances
from an added zero-edge source provide feasible potentials. Floyd--Warshall
therefore gives a fixed-radius decision procedure using $O(N^3)$ arithmetic
operations on the real edge costs. Numerical implementations still need
their own error control. The constraint set expands with $R$, justifying
bisection when suitably enclosed comparisons are available.

At the chain root the sum of all adjacent lower bounds is exactly $2\pi$.
Every feasible adjacent gap must therefore equal its lower bound. Thus
failure of one nonadjacent constraint at the cumulative-angle placement
cannot be repaired by redistributing gaps at that same root and order.
This is the mechanism of the seam obstruction.

\section{The seam criterion and the first three onsets}\label{sec:seam}
For integers $k\ge1$ and $n\ge k+2$, let $\sigma_{k,n}$ be the
rank cycle just defined, with rank $i$ replaced by $k+i-1$. Write
$R_{k,n}=\Rchain(\sigma_{k,n})$ and define its seam deficit
\begin{equation}\label{eq:seam-main}
 \Delta_{k,n}=\theta_{R_{k,n}}(n,k)+\theta_{R_{k,n}}(k,n-1)
                    -\theta_{R_{k,n}}(n,n-1).
\end{equation}
The neighbors of $k$ are $n-1,n$ in both parities.

\begin{theorem}[Fixed-order seam criterion and persistence]\label{thm:seam}
For every $k\ge1$, $n\ge k+2$, a placement in order $\sigma_{k,n}$
is feasible at $R_{k,n}$ if and only if $\Delta_{k,n}\ge0$.
In that case the cumulative-angle placement is feasible and
$\Rfull(\sigma_{k,n})=R_{k,n}$.

For every fixed $k$ there is a finite $s_k\ge4k+1$ such that the
deficit is negative for all $n\ge s_k$. For $k+2\le n<s_k$ it is
positive, except possibly for equality at $n=s_k-1$. For $k=1,2,3$
there is no integer equality case, and the full strict classifications
are given in Table~\ref{tab:onsets}.
\end{theorem}

\begin{table}[htbp]\centering
\caption{Complete chain-root classifications for the three specified values
of $k$. These are fixed-order results, not floating-circle thresholds.}
\label{tab:onsets}
\begin{tabular}{cccc}\toprule
$k$ & $\Delta_{k,n}>0$ & $\Delta_{k,n}<0$ & $s_k$\\\midrule
1 & $3\le n\le7$ & $n\ge8$ & 8\\
2 & $4\le n\le12$ & $n\ge13$ & 13\\
3 & $5\le n\le16$ & $n\ge17$ & 17\\\bottomrule
\end{tabular}
\end{table}

Appendix~\ref{app:seam} contains a complete proof, including both parity
edge lists, the two angular directions, the closing gap, small cycles,
equality cases and six rational endpoint bridges. The key inequality is
that every triangle defect is at least
$\delta_R=\theta_R(n,k)+\theta_R(k,n-1)-\theta_R(n,n-1)$.
A fan decomposition transfers this lower bound to every directed path.
At the root, the two complementary paths have total length $2\pi$,
which proves sufficiency; the forced adjacent gaps prove necessity.

Persistence follows by comparing the strictly increasing chain root with
a strictly decreasing algebraic threshold. This comparison, rather than
monotonicity of the raw deficit, proves eventual failure. We give no
explicit formula for $s_k$ when $k\ge4$ and do not exclude an equality
case there. The criterion itself includes equality.

\begin{corollary}\label{cor:small}
For $3\le n\le7$, $\Rstar(n)=R_{1,n}$ and the complete canonical
necklace is a global optimum. For every $n\ge8$, its cumulative-angle
placement at $R_{1,n}$ is infeasible.
\end{corollary}
\begin{proof}
Combine Theorems~\ref{thm:chain} and~\ref{thm:seam}. The second assertion
concerns the specified chain root, not a global optimizing contact graph.
\end{proof}

\section{Exact finite global certification}\label{sec:cert}
The finite certificate requires two independent arguments. For the lower
bound, we exclude every order at and slightly above $L_n$. For the upper bound,
we exhibit a feasible placement at $U_n$. The exact certificate data and
verifier are pinned in Section~\ref{sec:repro}. Historical floating-point
candidate generation,
Stage-A pruning and retained Top-K lists are not premises of either obligation.

\subsection{Necessary induced-cycle inequalities}
In any feasible placement, take any subset with at least three members
in its induced cyclic order. Its directed gaps sum to $2\pi$, and each
gap is at least the angle required for its endpoints. Consequently
\begin{equation}\label{eq:induced}
 \sum_{(a,b)\text{ consecutive in the induced cycle}}\theta_R(a,b)
 \le2\pi.
\end{equation}
This necessity holds even if a directed gap exceeds $\pi$. The verifier
tests the full order and the induced orders after deleting $1$ or
$\{1,2\}$, omitting subsets of size less than three. One strict violation
suffices to exclude a full order. We do not assume these tests are sufficient
for feasibility.

\subsection{Independent outward enclosures}
Let $S=2^{128}$. For each unordered pair at both endpoints the certificate
supplies integer bounds $\ell,u$ for $S\theta_R(a,b)$, with
$0<\ell\le u<3S$ and $u-\ell\le2$. Set
$q=ab/((R+a)(R+b))$, computed as a rational number. The verifier proves
\begin{equation}\label{eq:coscheck}
 \cos(\ell/S)>1-2q>\cos(u/S).
\end{equation}
Since $\cos$ is strictly decreasing on $(0,3)$, this places the true
angle strictly inside the proposed interval. No generator accuracy is assumed.

Cosine is enclosed by its Taylor series using outward integer rounding on
the internal grid $B=2^{224}$. Starting with the exact term $B$, multiply
successive absolute terms by $x^2/((2j+1)(2j+2))$, rounding the lower
bound down and the upper bound up. Signed addition uses the opposite bound
on a subtracted term. After 112 terms, the remaining terms decrease for
$0\le x\le4$; the first omitted term bounds the alternating remainder
with its proper sign. This argument does not require the early terms to
decrease. Appendix~\ref{app:arithmetic} gives the recurrence explicitly.

The circumference bounds $T_-,T_+$ satisfy $6S<T_-<T_+<7S$.
Opposite strict cosine signs at $T_-/(4S)$ and $T_+/(4S)$ enclose
the unique zero $\pi/2$ in $(1.5,1.75)$, proving
$T_-/S<2\pi<T_+/S$. All 908 pair intervals and this circumference
interval are independently checked. Use $w_{ab}=\ell_{ab}$ at $L_n$
as symmetric integer lower edge weights. A cycle cost strictly greater
than $T_+$ proves a strict violation of \eqref{eq:induced}.

\subsection{Complete canonical coverage and exact pruning}
For $n\le9$, pin the largest label first and retain one of each pair of
reflected permutations. Distinct labels give exactly $(n-1)!/2$ classes.
Every class receives an explicit induced-cycle violation.

For $n\ge10$, delete radius $1$. A full cyclic order becomes a skeleton
on $\{2,\ldots,n\}$; insertion into each of its $n-1$ gaps recovers
all full orders, including insertion in the closing gap. Deletion is the
inverse. Working modulo rotation and reflection preserves this bijection
between full classes and canonical skeletons with a chosen gap. With at
least three distinct skeleton labels, reflection fixes no oriented cycle.

Put $V=\{2,\ldots,n-1\}$ and anchor the skeleton at $n$. Compute the
forward dynamic program
\begin{align}
 D[\{j\},j]&=w_{nj},\nonumber\\
 D[M,j]&=\min_{i\in M\setminus\{j\}}
           \bigl(D[M\setminus\{j\},i]+w_{ij}\bigr).
 \label{eq:dp}
\end{align}
Induction on $|M|$ shows that this is the minimum weight of a path from
$n$ through $M$ ending at $j$. The root cycle cost is
$\min_j(D[V,j]+w_{jn})$. For a prefix of cost $c$ ending at $j$ with
unvisited set $M$, symmetry and path reversal give the exact minimum
completion cost $D[M\cup\{j\},j]$. Prune precisely when
\begin{equation}\label{eq:prune}
 c+D[M\cup\{j\},j]>T_+.
\end{equation}
Every true angular completion then violates the skeleton's necessary cycle
inequality. Equality is retained. At the root use the root cycle cost.
At a leaf the completion is the closing edge. All remaining choices are
visited without a candidate cap. The dynamic program (DP) has
$O(n2^n)$ states and
$O(n^2 2^n)$ arithmetic operations; traversal can still be factorial.

A pruned prefix with $h$ unvisited labels covers exactly $h!$ oriented
completions. If their total is $P$ and $K$ canonical skeletons remain,
the independently recomputed partition satisfies
\begin{equation}\label{eq:coverage}
 P+2K=(n-2)!,\qquad
 \frac{P(n-1)}2+K(n-1)=\frac{(n-1)!}2.
\end{equation}
The verifier also checks that $P$ is even, there are exactly $K$ reversed
leaves, and exactly $K(n-1)$ explicit insertions with distinct normalized
full-order classes. Every insertion must exhibit a strict induced-cycle
violation. The complete traversal and valid pruning justify coverage;
the factorial identity alone would not prove it. Trace hashes identify
recomputed DP, pruning and witness streams, rather than replacing them.

\subsection{Strictness for the infimum}\label{sec:strict}
Infeasibility at $L=L_n$ alone would not establish $L<\Rstar(n)$.
Here there is a smallest positive integer margin $m$ among all explicit
violations and pruning bounds. Every covered order has a necessary cycle
whose true angular excess at $L$ is at least $m/S$.
For $R\ge L$, differentiating \eqref{eq:kernel} gives the rational bound
\begin{align*}
 |\theta'_R(a,b)|
 &=\sqrt{\frac{q_R}{1-q_R}}
          \left(\frac1{R+a}+\frac1{R+b}\right)\\
 &\le d_{ab}:=\max\!\left(1,\frac{q_L}{1-q_L}\right)
          \left(\frac1{L+a}+\frac1{L+b}\right).
\end{align*}
We used decreasing $q_R$ and $\sqrt z\le\max(1,z)$ for $z\ge0$.
Let $D_n=n\max_{a<b}d_{ab}$ and
$\varepsilon_n=m/(2SD_n)>0$. Each relevant cycle has at most $n$
edges, so its violation persists on $[L,L+\varepsilon_n]$ with at
least half the margin. At smaller positive radii the angles are larger.
Therefore $\Rstar(n)\ge L+\varepsilon_n>L$. The verifier constructs
this positive rational buffer and checks $\varepsilon_n<U_n-L_n$.
No assertion about attainment is needed.

\subsection{Exact rational existence witnesses}
Each row contains at least one rational stereographic parameter $t_r$ for
every radius $r$, defining a center
\begin{equation}\label{eq:stereo}
 C_r=(U_n+r)\left(\frac{1-t_r^2}{1+t_r^2},
                         \frac{2t_r}{1+t_r^2}\right).
\end{equation}
Exact rational arithmetic checks $\|C_r\|^2=(U_n+r)^2$ and, for
every $a<b$, $\|C_a-C_b\|^2>(a+b)^2$. These centers alone prove
existence at $U_n$. Half-plane and cross-product comparisons check the
claimed polar order and distinct directions. Missing per-row witnesses
are rejected; checking an empty list would not prove existence.

The data additionally contain scaled angular positions. All directed and
wrap inequalities are checked with positive padding $\lceil S/10^{16}\rceil$,
using upper pair-angle bounds and the lower circumference bound. This is
constructive positive slack, not a negative error allowance. The angular
and rational Cartesian witnesses need not describe exactly the same angles.
There are 47 Cartesian witnesses in total. Their feasibility at $U_n$
does not identify the minimizing orders or their contact graphs at $\Rstar(n)$.

\section{Certified results for \texorpdfstring{$3\le n\le14$}{3 <= n <= 14}}\label{sec:results}
\begin{theorem}[Computer-certified finite brackets]\label{thm:brackets}
For each integer $n=3,\ldots,14$, the rational endpoints in
Table~\ref{tab:brackets} satisfy \eqref{eq:mainclaim}.
\end{theorem}
\begin{proof}[Computer-assisted proof]
The complete exact verifier checks the outward intervals, recomputes the
lower coverage in Section~\ref{sec:cert}, and validates at least one
rational Cartesian witness per row. The strict-margin argument in
Section~\ref{sec:strict} proves the open lower bound; the witnesses prove
the closed upper bound. All twelve endpoint differences equal $1/10^{11}$
as rational numbers. The pinned data, code, and complete execution command
are specified in Section~\ref{sec:repro}.
\end{proof}

\begin{table}[htbp]\centering
\caption{Exact global brackets. Every decimal is a rational endpoint;
the certified statement is $L_n<\Rstar(n)\le U_n$. The last column counts
supplied feasible upper witnesses, not exact optimal orders.}
\label{tab:brackets}
\begin{tabular}{rrrr}\toprule
$n$ & $L_n$ & $U_n$ & Witnesses\\\midrule
3 & 0.26086956521 & 0.26086956522 & 1 \\
4 & 0.84445358956 & 0.84445358957 & 1 \\
5 & 1.69549408120 & 1.69549408121 & 1 \\
6 & 2.79491951889 & 2.79491951890 & 1 \\
7 & 4.15318955374 & 4.15318955375 & 1 \\
8 & 5.76779428458 & 5.76779428459 & 1 \\
9 & 7.72672655261 & 7.72672655262 & 1 \\
10 & 9.97990738586 & 9.97990738587 & 4 \\
11 & 12.48872048718 & 12.48872048719 & 6 \\
12 & 15.25887043044 & 15.25887043045 & 9 \\
13 & 18.31756304721 & 18.31756304722 & 10 \\
14 & 21.66539518221 & 21.66539518222 & 11 \\
\bottomrule\end{tabular}

\end{table}

Table~\ref{tab:coverage} records the coverage sizes. For $n=10,11,12$
the minimum skeleton cycle already exceeds the upper circumference bound,
so the root DP comparison excludes every order without explicit insertions.
For $n=13,14$ some skeletons survive and all their insertions are tested.
Across the twelve cases, 268,648 full classes are explicitly checked and
3,374,988,556 are covered in total. These are finite combinatorial counts,
not evidence for any larger value of $n$.

\begin{table}[htbp]\centering
\caption{Complete lower coverage recomputed by the exact verifier.
An en dash indicates direct full-order enumeration. Zero explicit orders
at $n=10,11,12$ means the root skeleton bound excludes the whole domain.}
\label{tab:coverage}
\begin{tabular}{rrrr}\toprule
$n$ & Full classes covered & Explicit full classes & Retained skeletons\\\midrule
3 & 1 & 1 & -- \\
4 & 3 & 3 & -- \\
5 & 12 & 12 & -- \\
6 & 60 & 60 & -- \\
7 & 360 & 360 & -- \\
8 & 2,520 & 2,520 & -- \\
9 & 20,160 & 20,160 & -- \\
10 & 181,440 & 0 & 0 \\
11 & 1,814,400 & 0 & 0 \\
12 & 19,958,400 & 0 & 0 \\
13 & 239,500,800 & 2,952 & 246 \\
14 & 3,113,510,400 & 242,580 & 18,660 \\
\bottomrule\end{tabular}

\end{table}

The certificate does not assert exact decimal equality, uniqueness, or
a complete list of optimizing orders. Corollary~\ref{cor:small} gives
an additional analytic characterization for $n\le7$; it is logically
separate from the uniform bracket certificate.

\section{Interpreting finite floating regimes}\label{sec:floating}
At a specified placement, call circle $c$ \emph{strictly floating} if
every outer pair involving $c$ is strictly separated; its required tangency
to the central circle remains. An existential statement at the optimum
would assert that there exists an optimal placement with this property.
A universal statement would assert it for every optimal placement. Neither
follows from a nonessential constraint, a numerical slack, or a witness
at a radius above the optimum.

The public finite preprint and its saved numerical reconstructions describe
the pattern in Table~\ref{tab:regimes}. We retain it as interpretation of
those reconstructions. The floating entries for $n\ge8$ are numerical
descriptions, not additional exact consequences of Theorem~\ref{thm:brackets}.
For $3\le n\le7$, the displayed canonical optimum has every adjacent
pair tangent, so it has no strictly floating member; this is a statement
about that exhibited optimum.

\begin{table}[htbp]\centering\small
\caption{Finite geometric interpretation of the historical reconstructed
solutions (public finite preprint, arXiv:2607.28654v2). ``Paid'' and ``free''
describe comparison with a reduced-chain value. Except for the exhibited
necklace in the first row, these are numerical descriptions.}
\label{tab:regimes}
\begin{tabularx}{\textwidth}{llX}\toprule
$n$ & Reported floating set & Interpretation of the reconstruction\\\midrule
$3$--$7$ & $\varnothing$ & Complete canonical necklace; analytic global optimum.\\
$8,9$ & $\{1\}$ & Reduced necklace adjusted at a radius above its chain bound (paid placement).\\
$10$--$12$ & $\{1\}$ & Circle 1 accommodated at the reduced canonical chain value to reported numerical precision (free placement).\\
$13$ & $\{1\}$ & Second seam obstruction in the canonical order beginning at 2; a constrained reconstruction is used.\\
$14$ & $\{1,2\}$ & Reconstruction with two reported floating circles and a constrained remaining necklace.\\\bottomrule
\end{tabularx}
\end{table}

For an adjacent tangent pair $a,b$, the radius $\rho$ at which insertion
uses exactly their angular gap satisfies
$\theta_R(a,\rho)+\theta_R(\rho,b)=\theta_R(a,b)$.
Appendix~\ref{app:seam} derives the corresponding pocket curvature by
tangent addition, without importing a packing formula. Strict inequality
in the direction $\theta_R(a,\rho)+\theta_R(\rho,b)<\theta_R(a,b)$
permits slack against the two boundary circles. Strict floating in the
whole placement still requires all the other pair constraints.

The upper witnesses show why the radius must be specified: every outer
pair is strictly separated at $U_n$, so every outer circle is
strictly floating there under the placement-level definition. This does
not imply that every circle can float at $\Rstar(n)$; slack can disappear
as the radius decreases. Likewise, the exact seam onsets $8,13,17$ are
failures of three specified chain necklaces and are not thresholds at
which particular circles must float in globally optimal configurations.

\section{Limitations and open questions}
The exact global certificate stops at $n=14$. Extending its finite scope
requires new complete lower coverage and new feasible upper witnesses.
Numerical searches may guide those choices, but a best-known radius is only
an upper bound when feasibility is established. The present certificate
does not classify all optimizing orders, exact optimal contact graphs,
or the dimension of the set of optimal placements.

The most direct geometric question is to prove or refute the existential
floating descriptions at the exact optimum, with simultaneous strict slack
against every other outer circle. Universal floating assertions require
a separate argument. Fixed-$k$ persistent seam failure does not prove an
indefinite global floating cascade or repeated paid-then-free regimes.

The anti-Monge chain theorem extends to arbitrary distinct positive radii.
The integer-family seam onset statements do not automatically extend to
arbitrary sequences or deletions. Nor does any finite table establish an
asymptotic coefficient. The separate asymptotic manuscript mentioned in
the introduction treats that different question; no part of its proof is
imported into the present finite certification.

\section{Reproducibility statement}\label{sec:repro}
The verifier and its inputs are pinned at
\href{https://github.com/falker47/ringmin/tree/6c16af422d1cb38641c62d43b6e0e547921b9ba9}{\texttt{falker47/ringmin}},
commit
\begin{center}\small\texttt{\snapshot}\end{center}
From the root of that checkout, run the following command with Python 3.11
or later (standard library only):
\begin{verbatim}
python -I -S verify_global_brackets.py
\end{verbatim}
A successful run checks all twelve cases and their binding to the pinned
inputs. There is no partial-case or skip-frontier mode. The author's
recorded Python 3.14.3 run returned \texttt{PASS\_GLOBAL\_BRACKETS} for
908 angle intervals, 47 witnesses, 540 central tangencies, 3,004 outer
pairs and 6,008 angular inequalities.

Version \texttt{v1.1.0-dcg-presubmission} of the software and certificates
is archived on Zenodo~\cite{ringmin-archive} at release commit
\begin{center}\small\texttt{80919666c3c54f8ce20cf66c456d413f6a2c075f}\end{center}
The DOI \url{https://doi.org/10.5281/zenodo.22849826} identifies that
archived release. The source commit above identifies the verifier inputs
used here.

The certificate is the \texttt{certificate} field of
\begin{quote}\small\ttfamily\raggedright
reproducibility/global\_brackets/originals/source/\allowbreak
ringmin\_global\_interval\_candidate.json
\end{quote}
Its canonical serialization has SHA-256
\begin{center}\small\ttfamily
6f6c15e1db037a3faadadc52893a9a90ae7b5\\
9d6d5b42e934c81c1a0abd21cc0
\end{center}
The source commit's manifest gives the digests of the original files.

The verifier imports neither the production solver nor the original
generator. It checks angles with cosine enclosures in place of the
generator's arcsine/Machin evaluations, recomputes the lower dynamic
program in the forward direction, and parses the generator as static data
to bind the inputs. It derives from an independent reviewer's
implementation. A rerun reproduces that implementation. Source hashes
establish file identity but cannot authenticate historical executions,
platforms or timings.

Regression tests at the source commit reject false brackets, missing
witnesses, altered interval endpoints, collapsed directions, malformed
coverage and trace metadata, and changes to pinned inputs. Small
enumeration oracles test the dynamic program and closing-gap coverage.
The \href{https://github.com/falker47/ringmin/blob/4b3721df8061f1992cfd242df69b1a0643216d8f/paper_assets/journal_dcg/SOURCE_MAP.md}{source map}
locates the supporting material for individual claims.
The computer-assisted proof uses exact arithmetic. It has no
proof-assistant formalization, and both the mathematics and the verifier
remain subject to independent review.

\section*{Statements and Declarations}
\paragraph{Funding and competing interests.} The author received no
specific funding for this work and declares no relevant financial or
non-financial competing interests.
\paragraph{Author contribution.} Maurizio Falconi is the sole author.
\paragraph{Data and code availability.} Section~\ref{sec:repro} identifies
the certificate inputs, the verifier and the archived software release.
\paragraph{Use of artificial intelligence.} OpenAI ChatGPT/Codex assisted
with mathematical exploration, software development, computational
experiments, manuscript drafting and repository organization. The author
takes responsibility for the research, text, proofs and code.

\appendix
\section{Self-contained seam theorem and proof}\label{app:seam}
\subsection{Definitions and statement}

All surrounding circles are externally tangent to a central circle of radius
$R>0$. For surrounding radii $a,b>0$, define

\[
 \theta_R(a,b)=2\arcsin\sqrt{\frac{ab}{(R+a)(R+b)}}\in(0,\pi).
 \tag{A.1}\label{eq:seam1}
\]

If their centers have counterclockwise angular separation $A\in[0,2\pi]$,
the squared distance between them is
$(R+a)^2+(R+b)^2-2(R+a)(R+b)\cos A$. Thus their interiors are disjoint
if and only if

\[
 \theta_R(a,b)\le A\le2\pi-\theta_R(a,b).
 \tag{A.2}\label{eq:seam2}
\]

Indeed, the distance requirement is equivalent to
$\cos A\le1-2ab/((R+a)(R+b))=\cos\theta_R(a,b)$. In particular,
\textbf{both} directed arcs between a pair must have length at least its required
angle. Tangencies, including nonadjacent tangencies, are allowed.

Fix integers $k\ge1$, $n\ge k+2$, and put $N=n-k+1$. The
\textbf{canonical Supnick cycle} $\sigma_{k,n}$ is defined by the following
rank rule, which specifies every entry, without an order ellipsis.
Put $H=\lceil N/2\rceil$. For successive integers $j\ge0$:

\begin{itemize}
\item In a list $A$, append $1+2j$ if it is at most $H$, then append
$N-1-2j$ if it is greater than $H$.
\item In a list $B$, append $2+2j$ if it is at most $H$, then append
$N-2-2j$ if it is greater than $H$.
\end{itemize}

Stop each list when neither entry can be appended. Concatenate $A$, the
reverse of $B$, and the singleton $N$; replace each rank $i$ by
$k+i-1$, and join the last entry to the first. Low ranks partition
$1,\ldots,H$ by parity, high ranks partition $H+1,\ldots,N-1$, and
the appended rank is $N$. Hence this is a cycle through each radius once.
Rotations and reversal represent the same undirected cycle.
Subscripts on cycle entries are read modulo $N$.

Let $E_{k,n}$ be its undirected edge set, and define

\[
 C_{k,n}(R)=\sum_{(a,b)\in E_{k,n}}\theta_R(a,b),\qquad
 C_{k,n}(R_{k,n})=2\pi,\qquad R_{k,n}=R_{\rm chain}(\sigma_{k,n}).
 \tag{A.3}\label{eq:seam3}
\]

Existence and uniqueness of this positive root are proved below. At any
radius write

\[
 \delta_R=\theta_R(n,k)+\theta_R(k,n-1)-\theta_R(n,n-1),\qquad
 \Delta_{k,n}=\delta_{R_{k,n}}.
 \tag{A.4}\label{eq:seam4}
\]

A \textbf{fully feasible placement in this fixed order} consists of angles
$\phi_0<\cdots<\phi_{N-1}<\phi_0+2\pi$, with centers
$(R+s_i)(\cos\phi_i,\sin\phi_i)$, where
$(s_0,\ldots,s_{N-1})=\sigma_{k,n}$, such that \eqref{eq:seam2} holds for every
pair. The \textbf{cumulative-angle placement at the chain root} has $\phi_0=0$
and successive gaps $\theta_{R_{k,n}}(s_i,s_{i+1})$, including the
closing gap. Define $R_{\rm full}(\sigma)$ as the infimum of radii
admitting a fully feasible placement in order $\sigma$.

\textbf{Theorem (fixed-order seam criterion and persistence).} For every pair of
integers $k\ge1,n\ge k+2$:

\begin{enumerate}
\item The positive chain root exists uniquely. The following are equivalent:
$\Delta_{k,n}\ge0$; the cumulative-angle placement at $R_{k,n}$
is fully feasible; some fully feasible placement of this order exists at
$R_{k,n}$. When these hold,
$R_{\rm full}(\sigma_{k,n})=R_{k,n}$.
\item For each fixed $k$, there is a finite integer
$s_k\ge4k+1$ such that $\Delta_{k,n}<0$ for every $n\ge s_k$.
For $k+2\le n<s_k$ the deficit is positive, with at most one
possible exception: $\Delta_{k,s_k-1}=0$. Thus feasibility at the
chain root holds exactly before $s_k$, and failure persists thereafter.
\item For $k=1,2,3$, the complete classifications are:\nopagebreak[4]

\[\begin{array}{c|c|c}k & \Delta_{k,n}>0 & \Delta_{k,n}<0\\1&3\le n\le7&n\ge8\\2&4\le n\le12&n\ge13\\3&5\le n\le16&n\ge17\end{array}\]

In particular $s_1=8,s_2=13,s_3=17$, with no integer equality
case for these three values of $k$.
\end{enumerate}

The proof does not assert the absence of equality for every $k\ge4$,
nor give an explicit formula for their onsets. It establishes the weak
criterion in part 1 even when equality is not excluded.

\subsection{Exact edges, both parities, and growth of the chain root}

Put $L=k+n$. Reading the rank rule gives the following edge lists;
indexed families are listed in increasing order of $i$:

\[
\begin{array}{ll}
 N=2h:\quad &(k,n),\ (k+h-1,k+h),\\
 &(i,L-1-i),\quad k\le i\le k+h-2,\\
 &(i,L+1-i),\quad k+1\le i\le k+h-1;\\[2mm]
 N=2h+1:\quad &(k,n),\\
 &(i,L-1-i),\quad k\le i\le k+h-1,\\
 &(i,L+1-i),\quad k+1\le i\le k+h.
\end{array}
 \tag{A.5}\label{eq:seam5}
\]

Here $h\ge2$ in the even case and $h\ge1$ in the odd case. To
check the formula directly, consecutive low/high ranks inside either list
alternate endpoint sums $N$ and $N+2$. After translation these are
$L-1$ and $L+1$. The closing edge is $(k,n)$. At the join of
the two lists, an even-sized cycle has the additional central edge
$(k+h-1,k+h)$; an odd-sized cycle completes the two indexed families
instead. Their lower endpoints run through exactly the ranges in \eqref{eq:seam5}.
The counts are $2+(h-1)+(h-1)=2h$ and $1+h+h=2h+1$.
The endpoint ranges give distinct edges, so none is omitted or repeated.

For $N=3$ the cycle is $(k,k+1,k+2)$. For every $N\ge4$,
the rank list starts with $(1,N-1)$ and ends with $N$.
Consequently the two neighbors of $k$ are exactly $n-1,n$ in both
parities, including $N=3$.

For fixed positive $a,b$, \eqref{eq:seam1} is continuous and strictly decreasing in
$R$, with limits $\pi$ at zero and zero at infinity. Hence
$C_{k,n}$ decreases continuously from $N\pi>2\pi$ to zero.
This proves the asserted existence and uniqueness.

\textbf{Lemma A.1 (growth without an ordering-optimality premise).} For fixed
$k$, $R_{k,n+1}>R_{k,n}$, and $R_{k,n}\to\infty$.

\emph{Proof.} The kernel strictly increases in either surrounding radius, since
$a/(R+a)$ and $b/(R+b)$ do so separately in \eqref{eq:seam1}.
Compare \eqref{eq:seam5} at $n$ and $n+1$, so $L$ increases by one.

\begin{itemize}
\item If $N=2h+1$, match the seam and each indexed edge $(a,b)$ to
$(a,b+1)$. The new even list contains all these edges and the
additional central edge $(k+h,k+h+1)$.
\item If $N=2h$, match the seam and each indexed edge to $(a,b+1)$,
retaining the old central edge $(k+h-1,k+h)$ unchanged. That retained
edge is the last member of the new lower-sum family. The only remaining
edge in the new odd list is $(k+h,k+h+1)$.
\end{itemize}

Thus for every $R>0$, $C_{k,n+1}(R)>C_{k,n}(R)$. Strict decrease
in $R$ implies strict growth of the roots. No deletion of a vertex and
no optimality over other cycles is used.

All endpoints are at least $k$; at the root,
$2\pi\ge2N\arcsin(k/(R_{k,n}+k))$. Inversion on the increasing
branch, followed by $\sin x\le x$, gives

\[
 R_{k,n}\ge k\bigl(\csc(\pi/N)-1\bigr)
             \ge k(N/\pi-1)\longrightarrow\infty
 \quad (k\text{ fixed}). \tag{A.6}\label{eq:seam6}
\]

This completes the proof. \hfill$\square$

\subsection{The least triangle defect and the fan identity}

\textbf{Lemma A.2 (triangle defect, including equality).} For any three distinct
members $a,b,c$ of $\{k,\ldots,n\}$, at every $R>0$,

\[
 \theta_R(a,b)+\theta_R(b,c)-\theta_R(a,c)\ge\delta_R.
 \tag{A.7}\label{eq:seam7}
\]

Equality holds exactly when $b=k$ and $\{a,c\}=\{n-1,n\}$.

\emph{Proof.} Suppress $R$ from the notation. Differentiating \eqref{eq:seam1} on the
positive domain gives

\[
 \theta_1(a,b)=\frac{\sqrt{Rb/a}}{(R+a)\sqrt{R+a+b}}>0,\qquad
 \theta_{12}(a,b)=\frac{\sqrt R}{2\sqrt{ab}(R+a+b)^{3/2}}>0.
 \tag{A.8}\label{eq:seam8}
\]

Exchange endpoints so that $a<c$. Decreasing the middle argument
$b$ to $k$ gives
$\theta(a,b)+\theta(b,c)-\theta(a,c)\ge H(a,c)$, strictly if
$b>k$, where $H(x,z)=\theta(x,k)+\theta(k,z)-\theta(x,z)$.
Analytic values such as $\theta(k,k)$ are legitimate here; they are
not self-pair constraints. For $x,z\ge k$,

\[
 H_1(x,z)=\theta_1(x,k)-\theta_1(x,z)\le0,\qquad
 H_2(x,z)=\theta_2(k,z)-\theta_2(x,z)\le0.
\]

The inequalities are strict respectively when $z>k$ and $x>k$.
Since $a\le n-1$, $c\le n$, $c>k$, and $n-1>k$,
increasing first $a$ to $n-1$, then $c$ to $n$, proves
$H(a,c)\ge H(n-1,n)=\delta_R$. More explicitly the difference is

\[
 \int_a^{n-1}\int_k^c\theta_{12}(x,y)\,dy\,dx
 +\int_c^n\int_k^{n-1}\theta_{12}(x,y)\,dx\,dy.
 \tag{A.9}\label{eq:seam9}
\]

It vanishes exactly when $a=n-1,c=n$. Together with the strict
middle-argument comparison this proves the equality statement. This argument
covers every ordering of the three distinct radii. \hfill$\square$

For a simple path $P=(v_0,\ldots,v_m)$, define its angular slack as
the sum of its $m$ edge angles minus $\theta_R(v_0,v_m)$.
For $m\ge2$, cancellation gives the exact fan identity

\[
 S_R(P)=\sum_{j=1}^{m-1}
 [\theta_R(v_0,v_j)+\theta_R(v_j,v_{j+1})-\theta_R(v_0,v_{j+1})]
 \ge(m-1)\delta_R. \tag{A.10}\label{eq:seam10}
\]

All triples are distinct, so Lemma A.2 applies term by term. For $m=1$
the slack and the empty sum are both zero. If $\delta_R=0$ and
$m\ge3$, at least two distinct vertices serve as middle arguments;
they cannot both equal $k$. Thus the slack is strictly positive.
For $m=2$, equality in \eqref{eq:seam10} occurs only at the seam through $k$,
up to reversal. No triangle inequality has been assumed: \eqref{eq:seam7} permits a
negative lower bound.

\subsection{Full feasibility, closing gap and small cycles}

Work at $R=R_{k,n}$, and abbreviate $\Delta=\Delta_{k,n}$.
For two distinct positions $i<j$, the two cyclic paths are

\[
 P_+=(s_i,s_{i+1},\ldots,s_j),\qquad
 P_-=(s_j,s_{j+1},\ldots,s_{N-1},s_0,\ldots,s_i).
\]

Their edge counts are $d=j-i$ and $N-d$. Both are simple, even
when one uses the closing edge. If $\Delta\ge0$, applying \eqref{eq:seam10}
separately to both paths shows that each angular length is at least
$\theta_R(s_i,s_j)$. Their sum is exactly $2\pi$ by \eqref{eq:seam3}.
Hence \eqref{eq:seam2} holds for every pair in the cumulative-angle placement. It is
unnecessary to decide which direction is the shorter one.

For adjacent endpoints the one-edge slack is zero. Its complement has
exact slack $2\pi-2\theta_R(s_i,s_{i+1})>0$, since $\theta_R<\pi$.
This checks the upper constraint as well as the lower one, for the closing
pair just as for every other adjacent pair. Formula \eqref{eq:seam2} then proves the
claimed Cartesian non-overlap, and the center radii enforce tangency to the
central circle.

Conversely, suppose some placement of this order is feasible at the chain
root. Unroll its angles and let $g_i>0$ be all $N$ consecutive
gaps, with $g_{N-1}=\phi_0+2\pi-\phi_{N-1}$. Adjacent feasibility
gives $g_i\ge\theta_R(s_i,s_{i+1})$. Therefore

\[
 \sum_i[g_i-\theta_R(s_i,s_{i+1})]=2\pi-C_{k,n}(R)=0.
 \tag{A.11}\label{eq:seam11}
\]

Every summand is nonnegative, so every gap, including the closing gap, equals
its adjacent angle. The two-edge path $(n,k,n-1)$, which exists by
Appendix~A.2, consequently has length
$A=\theta_R(n,k)+\theta_R(k,n-1)$. Feasibility of its endpoints
requires $A\ge\theta_R(n,n-1)$, namely $\Delta\ge0$.
If $\Delta<0$, the complementary arc $2\pi-A$ also exceeds the
upper bound $2\pi-\theta_R(n,n-1)$. There is no freedom to repair
this failure by changing gaps at the same radius and order. This proves
all three equivalences in part 1.

At any feasible radius $r$ for this order, summing adjacent constraints
gives $C_{k,n}(r)\le2\pi$, hence $r\ge R_{k,n}$. Feasibility
at $R_{k,n}$ when $\Delta\ge0$ proves the asserted equality of
the full and chain radii.

For clarity the small cycles are explicit:

\begin{itemize}
\item For $N=3$, $\sigma=(k,k+1,k+2)$. All pairs are adjacent; each
two-edge complement has slack $2\pi-2\theta_R(a,b)>0$.
In particular $\Delta=2\pi-2\theta_R(k+2,k+1)>0$. There is no
nonadjacent pair and no equality case.
\item For $N=4$, $\sigma=(k,k+2,k+1,k+3)$. The pair $(k,k+1)$
has the two paths through $k+2$ and through $k+3$; the pair
$(k+2,k+3)$ has the two paths through $k$ and through $k+1$.
These are precisely the four nonadjacent directed constraints. Lemma A.2
bounds all four defects below by $\Delta$, with equality only on the
path through $k$. In particular the seam complement has two edges,
not three. Adjacent complements have the positive slack already computed.
Appendix~A.5 also gives $\Delta>0$, since $n=k+3\le4k$.
\end{itemize}

For any $N\ge4$, even a hypothetical $\Delta=0$ is sufficient:
only the two-edge seam can have zero slack among multi-edge paths, by \eqref{eq:seam10}
and its equality discussion. When $\Delta>0$, the least nonadjacent
directed-path slack is exactly $\Delta$, attained only by that seam
up to reversal. Longer paths have slack at least $2\Delta$; the other
two-edge paths have strictly larger defects by Lemma A.2.

\subsection{Algebraic threshold and eventual persistence}

We derive the threshold identity rather than invoke a circle-packing
formula. With $x=1/R$, $u=1/a$, $v=1/b$, \eqref{eq:seam1} gives

\[
 \tan\frac{\theta_R(a,b)}2
 =\sqrt{\frac{ab}{R(R+a+b)}}=\frac{x}{\sqrt{x(u+v)+uv}}.
 \tag{A.12}\label{eq:seam12}
\]

Set $q=\sqrt{x(u+v)+uv}>0$ and $w=x+u+v+2q>0$. The positive
square-root identities

\[
 \sqrt{xu+(x+u)w}=x+u+q,\qquad
 \sqrt{xv+(x+v)w}=x+v+q
\]

follow by squaring. Also

\[
 (x+u+q)(x+v+q)-x^2=q(2x+u+v+2q)>0.
\]

Consequently the sum of the half-angles for pairs $(a,1/w)$ and
$(1/w,b)$ has tangent $x/q$. Each half-angle lies in
$(0,\pi/2)$; the positive denominator means their sum also lies
there. It therefore equals $\theta_R(a,b)/2$, with no branch
ambiguity. We have proved

\[
 \theta_R(a,1/w)+\theta_R(1/w,b)=\theta_R(a,b).
 \tag{A.13}\label{eq:seam13}
\]

For the seam endpoints put

\[
 \alpha_n=\frac1n+\frac1{n-1},\quad \beta_n=\frac1{n(n-1)},\quad
 P_n(x)=x+\alpha_n+2\sqrt{\alpha_n x+\beta_n}.
 \tag{A.14}\label{eq:seam14}
\]

This is the Descartes pocket curvature, but \eqref{eq:seam13} proves every identity
needed here independently of that interpretation. Strict increase of the
two angles with the inserted radius gives, for every $R>0$,

\[
 \operatorname{sign}\delta_R
 =\operatorname{sign}(k-1/P_n(1/R))
 =\operatorname{sign}(P_n(1/R)-1/k). \tag{A.15}\label{eq:seam15}
\]

The function $P_n$ is continuous and strictly increasing for $x\ge0$,
and tends to infinity. Its value at zero is
$(1/\sqrt n+1/\sqrt{n-1})^2$, strictly decreasing in $n$.
For $k\ge1$, strict decrease of $t\mapsto1/\sqrt t$ gives
\[
 P_{4k}(0)=\left(\frac1{\sqrt{4k}}+\frac1{\sqrt{4k-1}}\right)^2
 >\left(\frac2{\sqrt{4k}}\right)^2=\frac1k,
\]
\[
 P_{4k+1}(0)=\left(\frac1{\sqrt{4k+1}}+\frac1{\sqrt{4k}}\right)^2
 <\left(\frac2{\sqrt{4k}}\right)^2=\frac1k,
\]
where squaring preserves the strict inequalities because all terms are
positive. Thus $P_{4k}(0)>1/k>P_{4k+1}(0)$, and $\delta_R>0$ for every
$R>0$ when $k+2\le n\le4k$; a unique positive crossing
$P_n(x)=1/k$ exists exactly when $n\ge4k+1$.

On this latter domain set $c=1/k$ and
$q_n=\sqrt{\alpha_n c+\beta_n}$. Squaring the crossing equation
produces $c+\alpha_n\pm2q_n$. The plus root is extraneous because
the unsquared equation requires $c-\alpha_n-x\ge0$, while the plus
root gives $-2(\alpha_n+q_n)<0$. Since the unique positive crossing
exists, it must be the minus root

\[
 \kappa_{k,n}=\frac1k+\frac1n+\frac1{n-1}
 -2\sqrt{\frac{2n+k-1}{kn(n-1)}}>0,\qquad
 T_{k,n}=\frac1{\kappa_{k,n}}. \tag{A.16}\label{eq:seam16}
\]

One can also check its branch directly: $P_n(0)<c$ gives
$q_n^2-\alpha_n^2=\alpha_n(c-\alpha_n)+\beta_n>0$, and
$\alpha_n\kappa_{k,n}+\beta_n=(q_n-\alpha_n)^2$; its positive
square root is $q_n-\alpha_n$. From \eqref{eq:seam15},

\[
 \operatorname{sign}\Delta_{k,n}
 =-\operatorname{sign}(R_{k,n}-T_{k,n})\quad(n\ge4k+1).
 \tag{A.17}\label{eq:seam17}
\]

For fixed $x\ge0$, both $\alpha_n$ and $\beta_n$ strictly
decrease with $n$, so $P_{n+1}(x)<P_n(x)$. Evaluating at the
crossing and using strict increase in $x$ shows
$\kappa_{k,n+1}>\kappa_{k,n}$, hence $T_{k,n+1}<T_{k,n}$.
Equation \eqref{eq:seam16} gives $T_{k,n}\to k$. Lemma A.1 now implies that
$D_{k,n}=R_{k,n}-T_{k,n}$ strictly increases to infinity on
$n\ge4k+1$.

Define $s_k$ as the first integer where $D_{k,n}>0$. It exists;
the deficits are negative from then on. Before it they are positive except
possibly at one integer with $D_{k,n}=0$. Such an equality, if present,
must immediately precede $s_k$ by strict increase. The earlier
no-threshold domain has strictly positive deficits. This proves part 2,
including every quantifier, without assuming monotonicity of the raw
angular deficits.

\subsection{Six exact bridges and the three onset classifications}

By \eqref{eq:seam17} and strict growth of $D_{k,n}$, it suffices to establish

\[
\begin{array}{ll}
 R_{1,7}<6<T_{1,7}, & T_{1,8}<51/10<R_{1,8},\\
 R_{2,12}<17<T_{2,12}, & T_{2,13}<14<R_{2,13},\\
 R_{3,16}<32<T_{3,16}, & T_{3,17}<32<R_{3,17}.
\end{array} \tag{A.18}\label{eq:seam18}
\]

Here is all the rational arithmetic needed to verify these bridges.

\subsubsection{Threshold sides, with signs checked before squaring}

For each row let $r$ be its separator, let
$h=1/k+\alpha_n-1/r>0$, and put
$v=4(\alpha_n/k+\beta_n)>0$. Then
$\kappa_{k,n}-1/r=h-\sqrt v$. The margins in Table~\ref{tab:seam2} are strictly
positive rational numbers:

\begin{table}[htbp]\centering\small
\caption{Exact threshold comparisons; all displayed square margins are positive.}
\label{tab:seam2}
\setlength{\tabcolsep}{5pt}
\begin{tabular}{lrrrll}
\toprule
$(k,n)$ & $r$ & $h$ & $v$ & Positive square margin & Conclusion \\
\midrule
$(1,7)$ & $6$ & $8/7$ & $4/3$ & $v-h^2=4/147$ & $r<T_{1,7}$ \\
$(1,8)$ & $51/10$ & $3061/2856$ & $8/7$ & $h^2-v=47737/8156736$ & $T_{1,8}<r$ \\
$(2,12)$ & $17$ & $1381/2244$ & $25/66$ & $v-h^2=239/5035536$ & $r<T_{2,12}$ \\
$(2,13)$ & $14$ & $643/1092$ & $9/26$ & $h^2-v=673/1192464$ & $T_{2,13}<r$ \\
$(3,16)$ & $32$ & $69/160$ & $17/90$ & $v-h^2=671/230400$ & $r<T_{3,16}$ \\
$(3,17)$ & $32$ & $691/1632$ & $3/17$ & $h^2-v=7465/2663424$ & $T_{3,17}<r$ \\
\bottomrule
\end{tabular}
\end{table}

All six rows lie in the positive-threshold domain, so reciprocation is
legitimate. No sign is inferred by squaring quantities of unknown sign.

\subsubsection{Chain sides, including every closing and central edge}

For each edge $e=(a,b)$ in the \textbf{list order of \eqref{eq:seam5}} put
$z_e^2=ab/((r+a)(r+b))$. In Table~\ref{tab:seam3}, the displayed integer
vector divided by $d$ defines the entire vector $(q_e)$.
An upper row certifies $q_e^2-z_e^2>0$ for every edge; a lower row
certifies $z_e^2-q_e^2>0$. All entries are positive. These are finite
rational comparisons, using \eqref{eq:seam5} and the displayed formula for $z_e^2$.

\begin{table}[htbp]\centering\small
\caption{Rational edge bounds, in the order of \eqref{eq:seam5}.}
\label{tab:seam3}
\setlength{\tabcolsep}{5pt}
\begin{tabularx}{\textwidth}{lrllX}
\toprule
$(k,n)$ & $r$ & Direction & $d$ & Numerators of $(q_e)$, in order \\
\midrule
$(1,7)$ & $6$ & upper & 100 & 28,27,34,37,37,41,43 \\
$(1,8)$ & $51/10$ & lower & 1000 & 316,466,307,390,428,414,462,487 \\
$(2,12)$ & $17$ & upper & 1000 & 209,204,236,257,270,276,250,274,291,301,306 \\
$(2,13)$ & $14$ & lower & 1000 & 245,348,240,278,304,320,330,291,320,340,353,361 \\
$(3,16)$ & $32$ & upper & 100 & 17,23,17,19,20,21,22,22,20,21,22,23,24,24 \\
$(3,17)$ & $32$ & lower & 100 & 17,16,18,20,21,22,22,22,19,21,22,23,24,24,24 \\
\bottomrule
\end{tabularx}
\end{table}

For example, the first edge of the first row is $(1,7)$, giving
$(28/100)^2-1/13=12/8125>0$. An even row includes its central
edge as the second entry. The exact minimum square margin in each row is,
respectively,

\[
 \frac{43}{30000},\quad\frac{4333}{26270000},\quad
 \frac{32}{359375},\quad\frac{3}{32500},\quad
 \frac{23}{70000},\quad\frac{1}{5625}.
 \tag{A.19}\label{eq:seam19}
\]

For upper rows we use two elementary estimates. For $0<z<1/2$,

\[
 \arcsin z<\frac{z}{\sqrt{1-z^2}}<\frac{2z}{\sqrt3}<\frac65z.
 \tag{A.20}\label{eq:seam20}
\]

The first comparison integrates the strictly increasing derivative; the
last follows by squaring positive quantities. For $0<u\le1/3$,

\[
 (1+3u^2/5)^2(1-u^2)-1
 =\frac{u^2(5-21u^2-9u^4)}{25}>0.
\]

Indeed the parenthesis decreases as a function of $u^2$ and equals
$23/9>0$ at $u^2=1/9$. Taking positive square roots and
integrating proves

\[
 \arcsin z<z+z^3/5\quad(0<z\le1/3). \tag{A.21}\label{eq:seam21}
\]

For lower rows simply use $\arcsin z>z$ on $(0,1)$. The standard
strict bounds $3<\pi<22/7$ suffice: the lower bound follows from the
inscribed regular hexagon, and the upper has the elementary witness

\[
 \frac{22}7-\pi=\int_0^1\frac{x^4(1-x)^4}{1+x^2}\,dx>0,
\]

obtained by polynomial division and $\int_0^1dx/(1+x^2)=\pi/4$.
Table~\ref{tab:seam4} gives an exact rational bound for every closure comparison.

\begin{table}[htbp]\centering\small
\caption{Exact chain comparisons for the six endpoint bridges.}
\label{tab:seam4}
\setlength{\tabcolsep}{5pt}
\begin{tabularx}{\textwidth}{lXl}
\toprule
$(k,n)$ & Bound on $C_{k,n}(r)/2=\sum_e\arcsin z_e$ & Strict comparison \\
\midrule
$(1,7)$ & $<(6/5)\sum q_e=741/250$, all $q_e<1/2$ & $<3<\pi$ \\
$(1,8)$ & $>\sum q_e=327/100$ & $>22/7>\pi$ \\
$(2,12)$ & $<\sum(q_e+q_e^3/5)=1457520693/500000000$, all $q_e<1/3$ & $<3<\pi$ \\
$(2,13)$ & $>\sum q_e=373/100$ & $>22/7>\pi$ \\
$(3,16)$ & $<\sum(q_e+q_e^3/5)=14885133/5000000$, all $q_e<1/3$ & $<3<\pi$ \\
$(3,17)$ & $>\sum q_e=63/20$ & $>22/7>\pi$ \\
\bottomrule
\end{tabularx}
\end{table}

Strict decrease of $C_{k,n}$ converts these to the root sides of \eqref{eq:seam18}.
Appendices~A.6.1 and A.6.2 therefore prove all six bridges without numerical
approximations to roots or transcendental functions.

\subsubsection{Completing the integer ranges}

For $k=1,2,3$, let $s$ be respectively $8,13,17$. The
two bridges for that $k$ give $D_{k,s-1}<0<D_{k,s}$. Strict
increase of $D_{k,n}$ on $n\ge4k+1$ supplies the same strict
signs throughout the respective earlier and later ranges. On
$k+2\le n\le4k$, Appendix~A.5 already gives strict positivity of
$\Delta$. Equation \eqref{eq:seam17} and the equivalence in Appendix~A.4 now prove
the entire table in part 3. No integer lies between the two endpoints of
each bridge, so no equality case remains for these three $k$.

\section{Outward integer cosine arithmetic}\label{app:arithmetic}
For completeness, the enclosure used in \eqref{eq:coscheck} can be
implemented using the following exact recurrence. Write $x^2=a/d$ with
$a\ge0,d>0$, and let $B=2^{224}$. The $j$th absolute series term
is $t_j=x^{2j}/(2j)!$. Start $l_0=u_0=B$, and set
\[
 l_{j+1}=\left\lfloor\frac{a l_j}{d(2j+1)(2j+2)}\right\rfloor,
 \qquad
 u_{j+1}=\left\lceil\frac{a u_j}{d(2j+1)(2j+2)}\right\rceil.
\]
Induction gives $l_j\le Bt_j\le u_j$ for every $j$. Form lower and
upper sums through $j=N-1$ by using $(l_j,u_j)$ for even $j$ and
$(-u_j,-l_j)$ for odd $j$. For $N=112$ and $0\le x\le4$, the
tail beginning with $t_N$ decreases, since
$t_{N+1}/t_N=x^2/((2N+1)(2N+2))<1$ and subsequent ratios are
smaller. As $N$ is even, the remainder lies in $[0,t_N]$: add $u_N$
only to the upper sum. Division by $B$ gives a rigorous interval for
$\cos x$. For an odd cutoff one would instead subtract $u_N$ from the
lower sum. All rounding is integer floor or ceiling, so no library
floating-point or transcendental error assumption enters this calculation.

For the actual circumference check the scaled endpoints are
\begin{align*}
T_-&=2138057168129933495719360746323741566601,\\
T_+&=2138057168129933495719360746323741566602.
\end{align*}
The interval checks require strict inequalities; an inconclusive enclosure
is a verification failure, not permission to infer a sign. Integer arithmetic
also evaluates every pruning margin exactly. Rational arithmetic evaluates
\eqref{eq:stereo}, the strict squared-distance comparisons and the buffers
in Section~\ref{sec:strict}. Thus analytic enclosure, combinatorial coverage
and geometric existence are separate, explicitly checked parts of the proof.

\end{document}